\documentclass{llncs}

\usepackage{graphicx}
\usepackage{dcolumn}
\usepackage{bm}

\begin{document}
\title{Generalized Buneman pruning for inferring the most parsimonious multi-state phylogeny }
\author{Navodit Misra\inst{1}  \and Guy Blelloch\inst{2} \and R. Ravi \inst{3} \and Russell Schwartz\inst{4}}

\institute{Department of Physics, Carnegie Mellon University, Pittsburgh, USA. \email{nmisra@andrew.cmu.edu} \\
\and Computer Science Department,  Carnegie Mellon University, Pittsburgh, USA. \email{guyb@cs.cmu.edu}\\
 \and Tepper School of Business, Carnegie Mellon University, Pittsburgh, USA. \email{ravi@cmu.edu} \\
 \and Department of Biological Sciences, Carnegie Mellon University, Pittsburgh, USA. \email{russells@andrew.cmu.edu} }

\maketitle

\begin{abstract}
Accurate reconstruction of phylogenies remains a key challenge in evolutionary biology. Most biologically plausible formulations of the problem are formally NP-hard, with no known efficient solution. The standard in practice are fast heuristic methods that are empirically known to work very well in general, but can yield results arbitrarily far from optimal.  Practical exact methods, which yield exponential worst-case running times but generally much better times in practice, provide an important alternative.  We report progress in this direction by introducing a provably optimal method for the weighted multi-state maximum parsimony phylogeny problem. The method is based on generalizing the notion of the Buneman graph, a construction key to efficient exact methods for binary sequences, so as to apply to sequences with arbitrary finite numbers of states with arbitrary state transition weights. We implement an integer linear programming (ILP) method for the multi-state problem using this generalized Buneman graph and demonstrate that the resulting method is able to solve data sets that are intractable by prior exact methods in run times comparable with popular heuristics.  Our work provides the first method for provably optimal maximum parsimony phylogeny inference that is practical for multi-state data sets of more than a few characters.
\end{abstract}

\section*{Introduction}
One of the fundamental problems in computational biology is that of inferring evolutionary relationships between a set of observed amino acid sequences or taxa.  These evolutionary relationships are commonly represented by a tree (phylogeny) describing the descent of all observed taxa from a common ancestor, a reasonable model provided we are working with sequences over small enough regions or distant enough relationships that we can neglect recombination or other sources of reticulation~\cite{Posada}. Several criteria have been implemented in the literature for inferring phylogenies, of which one of the most popular is maximum parsimony (MP). Maximum parsimony defines the tree(s) with the fewest mutations as the optimum, generally a reasonable assumption for short time-scales or conserved sequences. It is a simple, non-parametric criterion, as opposed to common maximum likelihood models or various popular distance-based methods \cite{Felsenstein}. Nonetheless, MP is known to be NP-hard~\cite{Foulds} and practical implementations of MP are therefore generally based on heuristics which do not guarantee optimal solutions.

For sequences where each site or character is expressed over a set of discrete states, MP is equivalent to finding a minimum Steiner tree displaying the input taxa. For example, general DNA sequences can be expressed as strings of four nucleotide states and proteins as strings of 20 amino acid states.
Recently,  Sridhar {\it et al.} \cite{Sri} used integer linear programming to efficiently find global optima for the special case of sequences with binary characters, which are important when analyzing single nucleotide polymorphism (SNP) data. The solution was made tractable in practice in large part by a pruning scheme proposed by Buneman and extended by others~\cite{Buneman,Barthelemy,Bandelt}. The so-called Buneman graph $\mathcal{B}$ for a given set of observed strings is an induced sub-graph of the complete graph $\mathcal{G}$ (whose nodes represent all possible strings of mutations) such that $\mathcal{B}\subseteq\mathcal{G}$ still contains all distinct minimum Steiner trees for the observed data.   By finding the Buneman graph, one can often greatly restrict the space of possible solutions to the Steiner tree problem.  While there have been prior generalizations of the Buneman graph to non-binary characters~\cite{BandeltQM,Huber}, they do not provide any comparable guarantees usable for accelerating Steiner tree inference.

In this paper, we provide a new generalization of the definition of Buneman graph for any finite number of states that guarantees the resulting graph will contain all distinct minimum Steiner trees of the multi-state input set. Further, we allow transitions between different states to have independent weights. We then utilize the integer linear programming techniques developed in \cite{Sri} to find provably optimal solutions to the multi-state MP phylogeny problem. We validate our method on four specific data sets chosen to exhibit different levels of difficulty:  a set of nucleotide sequences from {\em Oryza rufipogon} \cite{rice}, a set of human mt-DNA sequences representing prehistoric settlements in Australia \cite{Ychr}, a set of HIV-1 reverse transcriptase amino acid sequences and, finally, a 500 taxa human mitochondrial DNA data set. We further compare the performance of our method, in terms of both accuracy and efficiency, with leading heuristics, {\tt PAUP*}~\cite{PAUP} and the {\tt pars} program of PHYLIP~\cite{Phylip}, showing our method to yield comparable and often far superior run times on non-trivial data sets.

\section*{Methods}

\subsection*{Notation \& Background}
Let $H$ be an input matrix that specifies a set of $N$ taxa $\chi$, over a set of $m$ characters $C=\{c_1, \ldots c_m\}$ such that $H_{ij}$ represents the $j^{th}$ character of the $i^{th}$ taxon. The taxa of $H$ represent the terminal nodes of the Steiner tree inference.  Further, let $n_k$ be the number of admissible states of the $k^{th}$ character $c_k$. The set of all possible states is the space $\mathcal{S} \equiv \{0,1,\ldots n_1-1\}\bigotimes \ldots \bigotimes \{0,1,\ldots n_m-1\}$. We will represent the $i^{th}$ character of any element $b\in S$, by $(b)_i$. The state space $\mathcal{S}$ can be represented as a graph $\mathcal{G}=(V_\mathcal{G},E_\mathcal{G})$ with the vertex set $V_\mathcal{G} = \mathcal{S}$ and edge set $E_\mathcal{G} = \{(u,v)|u,v \in \mathcal{S}, \sum_{c_p\in C}^{m}\delta[(u)_p,(v)_p]=1\}$, where $\delta[a,b] =0$ if $a=b$ and 1 otherwise. Furthermore, let $\bm{\alpha}=\{\alpha_{p}|c_p\in C \}$ be a set of weights, such that $\alpha_{p}[i,j]$ represents an edge length for a transition between states $i,j \in \{0,\ldots n_p -1\}$ for character $c_p$. We will assume that these lengths are positive (states that share zero edge length  are indistinguishable), symmetric in $i,j$ and satisfy the triangle inequality.
\begin{equation}
\label{eq:triangle}
\alpha_{p}[i,j] + \alpha_{p}[j,k] \geq \alpha_{p}[i,k] \quad \forall \quad  i,j,k\in\{0, \ldots n_p -1\}
\end{equation}
Non-negativity and symmetry are basic properties for any reasonable definition of length. If a particular triplet of states (say $i,j,k$) does not satisfy the triangle inequality in equation~\ref{eq:triangle}, we can set $ \alpha_{p}[i,k] = \alpha_{p}[i,j] + \alpha_{p}[j,k]$ and still ensure that the shortest path connecting any set of states remains the same. We can now define a distance $d_\alpha$ over $\mathcal{G}$, such that for any two elements $u,v\in V_{\mathcal{G}}$
\begin{equation}
\label{eq:dist}
d_{\bm{\alpha}}[u,v] \equiv \sum_{p\in C}^{m} \alpha_{p}[(u)_p,(v)_p]
\end{equation}

Given any subgraph $K=(V_K,E_K)$ of $\mathcal{G}$, we can define the length of $K$ to be the sum of the lengths of all the edges $L(K) \equiv \sum_{(u,v)\in E_K} d_{\bm{\alpha}}[u,v]$.
The maximum parsimony phylogeny problem for $\chi$ is equivalent to constructing the minimum Steiner tree $T_*$ displaying the set of all specified taxa $\chi$, i.e., any tree $T_* (V_* , E_* )$ such that $\chi \subseteq V_* $ and $ L(T_*)$ is minimum. Note that $T_*$ need not be unique.
\subsection*{Pre-processing}
Before we construct the generalized Buneman graph corresponding to an input, we perform a basic pre-processing of the data. The set of taxa in the input $H$ might not all be distinct over the length of sequence represented in $H$. These correspond to identical rows in $H$ and are eliminated. Similarly, characters that do not mutate for any taxa do not affect the true phylogeny and can be removed.  Furthermore, if two characters are expressed identically in $\chi$ (modulo a relabeling of the states), we will represent them by a single character with each edge length replaced by the sum of the edge lengths of the individual characters. In case there are $n$ such non-distinct characters, one of them is given edge lengths equal to the sum of the corresponding edges in each of the $n$ characters and the rest are discarded. These basic pre-processing steps are often useful in considerably reducing the size of input.

\subsection*{Buneman graph}
The Buneman graph was introduced as a pruning of the complete graph for the special case of binary valued characters. For this special case it is useful to introduce the notion of binary splits $c_p(0)|c_p(1)$ for each character $c_p\in C$, which partition the set of taxa $\chi$ into two sets $c_p(0)$ and $c_p(1)$ corresponding to the value expressed by $c_p$. Each of these sets is called a block of $c_p$. Each vertex of the Buneman graph $\mathcal{B}$ can be represented by an $m$-tuple of blocks $[c_1(i_1),c_2(i_2), \ldots, c_m(i_m)]$, where $i_j = 0 $ or 1, for $j \in\{1, 2, \ldots m\}$. To construct the Buneman graph, a rule is defined for discarding/retaining the subset of vertices contained in each pair of overlapping blocks $[c_p(i_p),c_q(i_q)]$ for each pair of characters $(c_p, c_q)\in C\times C$. All vertices which satisfy $c_p(i_p)\cap c_q(i_q)=\emptyset$ for any pair of characters $(c_p,c_q)$ can be eliminated, while those for which $c_p(i_p)\cap c_q(i_q) \ne \emptyset$ for all $[c_p(i_p), c_q(i_q)]$ are retained.  Buneman previously established for the binary case that the retained vertex set will contain all terminal and Steiner nodes of all distinct minimum length Steiner trees.

We extend this prior result to the weighted multi-state case by presenting an algorithm analogous to the binary case to construct a graph with these properties.

\subsection*{Algorithm for constructing the generalized Buneman graph}\label{BunAlgo}
Briefly, the algorithm looks at the input matrix projected onto each distinct pair of characters $p,q$ and constructs a $n_p\times n_q$ matrix $C(p,q)$, where the $i \times j^{th}$ element $C(p,q)_{ij}$ is 1 only if there is at least one taxon $t$ such that $(t)_p =i$ and $(t)_q=j$ and zero otherwise. The algorithm then implements a rule for each such pair of characters $p,q$ that allows us to enumerate the possible states of those characters in any optimal Steiner tree. For clarity, we will assume that each state for each character is expressed in at least one input taxon, since states that are not present in any taxa cannot be present in a minimum length tree because of the triangle inequality. The rule is defined by a $n_p\times n_q$ matrix $R(p,q)$ determined by the following algorithm :
\begin{enumerate}
\item $R(p,q)_{ij} \leftarrow C(p,q)_{ij}$ for all $i\in\{0,1,\ldots n_p-1\}$ and $j\in\{0,1,\ldots n_q-1\}$.
\item If all non-zero entries in $C(p,q)$ are contained in the set of elements \begin{displaymath}\left( \cup_{k} C(p,q)_{ik} \right)  \bigcup \left( \cup_{k} C(p,q)_{kj} \right)\end{displaymath} for a unique pair $i\in\{0,1,\ldots n_p-1\}$ and $j\in\{0,1,\ldots n_q-1\}$ then $R(p,q)_{xy}\leftarrow1$ for all $x,y$ such that either $x = i$ or $y = j$ (See Fig~\ref{Pruning} where this pair of states are denoted $i_{pq}$ and $i_{qp}$.)
\item If the condition in step 2 is not satisfied then set $R(p,q)_{ij}\leftarrow 1$ for all $i,j$.
\end{enumerate}
\begin{figure}[t!]
\begin{center}
\includegraphics[scale=0.5]{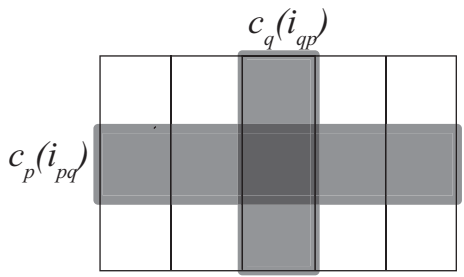}
\caption{An example of the generalized Buneman pruning condition. If all taxa in $\chi$ are present in the shaded region, vertices in all other blocks can be discarded.}
\label{Pruning}
\end{center}
\end{figure}
This set of rules $\{R\}$ then defines a subgraph $B_{pq}\subseteq\mathcal{G}$ for each pair of characters $p,q$, such that any vertex $v \in B_{pq}$ if and only if $R(p,q)_{(v)_p(v)_q} =1$. The intersection of these subgraphs $\mathcal{B}= \cap_{c_p,c_ q\in C}B_{pq}$ then gives the generalized Buneman graph for $\chi$ given any set of distance metrics $\bm{\alpha}=\{\alpha_p| c_p\in C\}$. Note that the Buneman graph of any subset of $\chi$ is a subset of $\mathcal{B}$. It is easily verified that for binary characters, our algorithm yields the standard Buneman graph.

The remainder of this paper will make two contributions.  First, it will show that the generalized Buneman graph $\mathcal{B}$ defined above contains all minimum Steiner trees for the input taxa $\chi$. This will in turn establish that restricting the search space for minimum Steiner trees to $\mathcal{B}$ will not affect the correctness of the search.  The paper will then empirically demonstrate the value of these methods to efficiently finding minimum Steiner trees in practice.

Before we prove that all Steiner minimum trees connecting the taxa are displayed in $\mathcal{B}$, we need to introduce the notion of a \emph{neighborhood decomposition}. Suppose we are given any tree $T(V,E)$ displaying the set of taxa $\chi$. We will contract each degree-two Steiner node (i.e., any node that is not present in $\chi$) and replace its two incident edges by a single weighted edge. Such trees are called {\it X-Trees}~\cite{Semple}. Each X-Tree can be uniquely decomposed into its \emph{phylogenetic X-Tree} components, which are maximal subtrees whose leaves are taxa. Formally, each phylogenetic X-Tree $P(\psi)$ consists of a set of taxa $\psi \subseteq \chi$ and a tree displaying them, such that there is a bijection or labeling $\eta: l_P \rightarrow \psi$ between elements of $\psi$ and the set of leaves $l_P\in P(\psi)$~\cite{Semple} (Fig \ref{SKLTN}) . All vertices in $P(\psi)$ with degree 3 or higher will be called \emph{branch points}. From now on we will assume that given any input tree, such a decomposition has already been performed (Fig \ref{SKLTN}).
Two phylogenetic X-Trees $P(\psi)$ and $P'(\psi)$ are considered \emph{equivalent} if they have identical length and the same tree topology. By identical tree topology, we mean there is a bijection between the edge set of the two trees, such that removing any edge and its image partitions the leaves into identical bi-partitions. We define two trees to be \emph{neighborhood distinct} if after neighborhood decomposition they differ in at least one phylogenetic X-Tree component. We define a labeling of the phylogenetic X-Tree as an injective map $\Gamma: P \rightarrow \mathcal{G}$ between the vertices of $P(\psi)$ and those of the graph $\mathcal{G}$ such that $\Gamma_u$ represents the character string for the image of vertex $u$ in $\mathcal{G}$. Since leaf labels are fixed to be the character strings representing the corresponding taxa,  $\Gamma_t = \eta_t \in \psi$ for any leaf $t \in l_P$. Identical phylogenetic X-Trees can, however, differ in the labels $\Gamma_u$ of internal branch points $u\in P\setminus l_P$.

\begin{figure}[t!]
\begin{center}
\includegraphics[scale=0.5]{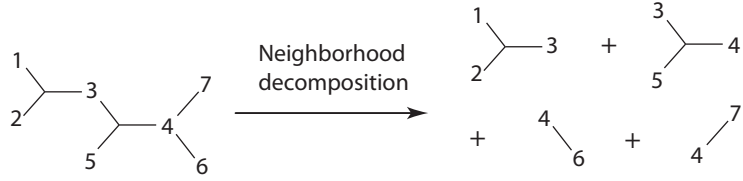}
\caption{An input tree and its phylogenetic X-Tree components,with taxa labelled by integers.}
\label{SKLTN}
\end{center}
\end{figure}
We will use a generalization of the Fitch-Hartigan algorithm to weighted parsimony proposed by Erdos and Szekely~\cite{Erdos,Jiang}. The algorithm uses a similar forward pass/backward pass technique to compute an optimal labeling for any phylogenetic X-Tree $T(\psi)$. Arbitrarily root the tree $T(\psi)$ at some taxon $\zeta$ and starting with the leaves compute the minimum length $minL({\Gamma}_b,T_b)$ of any labeling of the subtree $T_b$ consisting of the vertex $b$ and its descendants, where the root $b$ is labeled ${\Gamma}_b$ as follows.
\begin{enumerate}
\item If $\Gamma_b$ labels a leaf $\eta_b\in \psi$, $minL(\Gamma_b =\eta_b,T_b) =0$ and $\infty$ otherwise.
\item If $b$ has $k$ children $D_b=\{v_1, \ldots v_k\}$, and $ T_v$ is the subtree consisting of $v \in D_b$ and its descendants,
\begin{equation}\label{Score}
minL(\Gamma_b,T_b)= \sum_{v\in D_b}\min_{\Gamma_v} \{ minL(\Gamma_v,T_v) + d_{\bm{\alpha}}[\Gamma_b,\Gamma_v] \}
\end{equation}
where the minimum is to be taken over all possible labels $\Gamma_v$ for each character and for each child $v\in D_b$.
\end{enumerate}
The optimal labeling of $T(\psi)$ is one which minimizes the length at the root: $L(T)=minL(\eta_\zeta,T_\zeta)$.  Labels for each descendant are inferred in a backward pass from the root to the leaves and using equation~\ref{Score}. Note that the minimum length of a tree is just the sum of minimum lengths for each character, i.e., $minL(\Gamma_b,T_b)=\sum_{c_s\in C}minL(\Gamma_b,T_b)^{(s)}$,  where $minL(\Gamma_b,T_b)^{(s)}$ is the minimum cost of tree $T_b$ rooted at $b$ for character $c_s$.

Briefly, our proof is structured as follows: Given any phylogenetic X-Tree $T(\psi)$ labeling (typically denoted $\Gamma$ below), we will show that the generalized Buneman pruning algorithm for each pair of characters $(c_p, c_q)$ defines a subgraph $B_{pq}$ which contains at least one possible labeling of no higher cost (typically denoted $\Phi$ below) for $T(\psi)$.  We will then show that the intersection of these subgraphs $\mathcal{B}= \cap_{p\neq q}B_{pq}$ thus contains an optimal labeling for $T(\psi)$.

If the pruning condition in step 2 of the algorithm that defines the Buneman graph is not implemented for the pair of characters $(c_p,c_q)$, then $B_{pq}=\mathcal{G}$ and all labels are necessarily inside $B_{pq}$. We prove the following lemma for the case when the pruning condition is satisfied, ie., there exist unique states $i_{pq}$ of $c_p$ and $i_{qp}$ of $c_q$, such that each element in the set of leaves $l_T=\{t\in T(\psi) |\eta_t\in \psi\}$ either has $(\eta_t)_p=i_{pq}$ or $(\eta_t)_q=i_{qp}$ or both.  Each time we relabel vertices, we will keep all characters except $c_p$ and $c_q$ fixed. To economize our notation, we will represent the sum of costs in $c_p$ and $c_q$ of the tree $T$ labeled by $\Gamma$, which has some branch point $b$ as the root, simply by writing $L(\Gamma,T) = L(\Gamma,T)^{(p)} + L(\Gamma,T)^{(q)}$. We use the notation $\Gamma_x=[(\Gamma_x)_p,(\Gamma_x)_q]$ to represent the label for a vertex $x$ and suppress the state of all other characters.
\begin{figure}[t!]
\begin{center}
\includegraphics[scale=0.35]{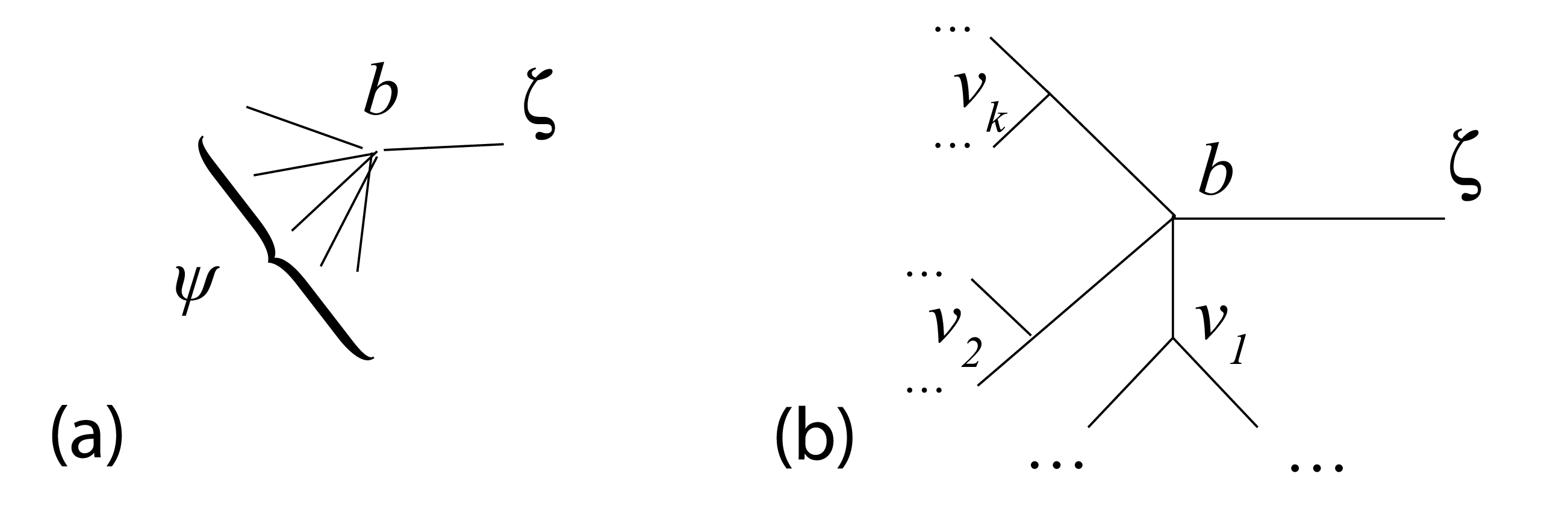}
\caption{(a) The base case of a degree $|\psi|$ star that can be attached to a parent vertex $\zeta$ in the Erdos-Szekely algorithm. (b) $T(\psi)$ for the general case (see Lemma~\ref{BpqLem})}
\label{BpqInduction}
\end{center}
\end{figure}
\begin{lemma}\label{BpqLem} Given any phylogenetic X-Tree $T(\psi)$ with $\psi \subseteq B_{pq}$, and a  labeling $\Gamma$, such that an internal branch point $b \in T\setminus l_T$ is labeled outside $B_{pq}$, i.e., $\Gamma_b\notin B_{pq}$, there exists an alternate labeling $\Phi$ of $T(\psi)$ inside $B_{pq}$ such that
\begin{enumerate}
\item either $L(\Gamma,T)  \geq L(\Phi,T) + d_{\bm{\alpha}}[\Gamma_b,\Phi_b]  $, or ---
\item $L(\Gamma,T) \geq  L(\Phi,T)$  for each of the following choices: $\Phi_b=[i_{pq},i_{qp}]$ or $[i_{pq},(\Gamma_b)_q]$ or $[(\Gamma_b)_p,i_{qp}] $, and $\Phi_v = \Gamma_v$ for all $v \neq b$. We will call a tree that satisfies this second condition a $(c_p,c_q)$-Tree
\end{enumerate}
\end{lemma}

\begin{proof}
We will use induction on the number of internal branch points outside $B_{pq}$ to prove the claim. Without loss of generality we can consider all branch points of $T(\psi)$ to be labeled outside $B_{pq}$.  If some branch points are labeled inside $B_{pq}$ then they can be treated as leaves of smaller X-Tree(s) that have all branch points outside $B_{pq}$. This is similar to the neighborhood decomposition we performed earlier for those branch points that were present in the set of input taxa. The set of branch points is then the set $T\setminus l_T=\{u\in T|\Gamma_u\notin B_{pq}\}$.

For the base case assume all the leaves are joined at a single branch point $b$ to form a star of degree $|\psi|$ (see Fig.~\ref{BpqInduction}(a) without the root $\zeta$). We can group the leaves into three sets:
 \begin{enumerate}
 \item $I=\{ \eta_u=[i_{pq},y_u]|y_u\neq i_{qp}, \eta_u\in\psi\}$
 \item $II=\{\eta_v=[x_v,i_{qp}]|x_v\neq i_{pq}, \eta_v\in\psi\}$
 \item $III=\{\eta_w=[i_{pq},i_{qp}]| \eta_w\in\psi\}$
 \end{enumerate}
 The cost of the tree for $c_p$ and $c_q$, with branch point $\Gamma_b=[x,y]$, is
 \begin{eqnarray}
L(\Gamma,T)^{(p)}+ L(\Gamma,T)^{(q)}&=&\sum_{u\in I} (\alpha_{p}[x,i_{pq}] +\alpha_{q}[y,y_u] ) +\sum_{v\in II} (\alpha_{p}[x,x_v] \nonumber \\
 &+& \alpha_{q}[y,i_{qp}] ) + \sum_{w\in III} (\alpha_{p}[x,i_{pq}] +\alpha_{q}[y,i_{qp}] )
 \end{eqnarray}
 The only way for $L(\Gamma,T)^{(p)}+ L(\Gamma_b,T)^{(q)}$ to be minimum with $x\neq i_{pq}$ and $y\neq i_{qp}$, is if $III=\emptyset$ and $|I| =|II|$. For contradiction, suppose $|I|+|III|>|II|$.  We could then define a labeling $\Phi$ identical to $\Gamma$ over all characters, except $\Phi_b=[i_{pq},y]$, such that $d_{\bm{\alpha}}[\Gamma_b,\Phi_b]=\alpha_p[\Gamma_b,\Phi_b]$.  We could then reduce the length, since
 \begin{eqnarray}  L(\Gamma,T)^{(p)} &=& \sum_{u\in I} \alpha_{p}[x,i_{pq}]  + \sum_{v\in II} \alpha_{p}[x,x_v]  +\sum_{w\in III} \alpha_{p}[x,i_{pq}] \nonumber \\
 &\geq&  \alpha_{p}[x,i_{pq}]+\sum_{v\in II} (\alpha_{p}[x,x_v] +  \alpha_{p}[x,i_{pq}])  \nonumber \\
 &\geq& \alpha_{p}[x,i_{pq}] + \sum_{v\in II} \alpha_{p}[i_{pq},x_v] = L(\Phi,T)^{(p)} +d_{\bm{\alpha}}[\Gamma_b,\Phi_b]\end{eqnarray}
where the last inequality follows from the triangle inequality. Similarly, if $|II|+|III|>|I|$, we could define $\Phi_b =[x,i_{qp}]$ and arrive at $  L(\Gamma,T)^{(q)}\geq  L(\Phi,T)^{(q)} + d_{\bm{\alpha}}[\Gamma_b,\Phi_b] $.

On the other hand if $|I|=|II|$ and $III=\emptyset$ setting $\Phi_b=[i_{pq},y]$ or $\Phi_b=[x,i_{qp}]$ or $\Phi_b=[i_{pq},i_{qp}]$ all achieve a length no more than $L(\Gamma,T)^{(p)}+ L(\Gamma,T)^{(q)}$. Therefore, this is a $(c_p,c_q)$-Tree. This proves the base case for our proposition.

We will now assume that the claim is true for all trees with $n$ branch points or less. Suppose we have a labeled tree $T(\psi)$ with $n+1$ branch points which are all outside $B_{pq}$. Let $D_b=\{v_1, \dots v_k\}$ be the children of a branch point $b$ in $T(\psi)$ and $\{T_{1},\ldots T_{k}\}$ be the subtrees of each $v\in D_b$ and their descendants. Note that some of these descendants may be leaves. Since $T(\psi)$ has at least two branch points, one of its descendants (say $v_1$) must be a branch point (Fig~\ref{BpqInduction}(b)). Let $T_b =T\setminus T_{1}$ be the subtree consisting of $b$ and all its other descendants.  For clarity we will use the notation $\Gamma_b=[x_b,y_b]$ and $\Gamma_{v_1}=[x_1,y_1]$. This implies,
\begin{eqnarray}\label{eq:gamma}
L(\Gamma,T) &=&L(\Gamma,T_b) + L(\Gamma,T_1) + d_{\bm{\alpha}}[\Gamma_{b},\Gamma_{v_1}]\nonumber \\
&=&L(\Gamma,T_b) + L(\Gamma,T_1) + \alpha_p[x_b,x_1] + \alpha_q[y_b,y_1]
\end{eqnarray}
There are four possibilities.
\begin{enumerate}
\item Both $T_b$ and $T_{1}$ are $(c_p,c_q)$-Trees with $n$ or less branch points - In this case, by induction, both $T_b$ and $T_1$ can be relabeled with $\Phi_b$ and $\Phi_{v_1}$ of the form  $[i_{pq},i_{qp}]$. Since the cost in $c_p$ and $c_q$ of the edge $(b,v_1)$ is now zero, we have an optimal labeling of $T(\psi)$ within $B_{pq}$ and $L(\Gamma,T) \ge L(\Phi,T)$ . Note that each of the choices of the form $[i_{pq},y_1]$ or $[x_1,i_{pq}]$ for relabeling of $b$ also satisfy property 2 of the claim. Therefore, this is a $(c_p,c_q)$-Tree.
\item $T_b$ is a $(c_p,c_q)$-Tree, but $T_1$ is not. Therefore, there is a labeling $\Phi$ of $T_1$ with either $\Phi_{v_1}=[i_{pq},y_1]$ and/or  $\Phi_{v_1}=[x_1,i_{pq}]$ such that
\begin{equation}
L(\Gamma,T_1) \geq L(\Phi,T_1) + d_{\bm{\alpha}}[\Gamma_{v_1},\Phi_{v_1}]
\end{equation}
 Let us assume for concreteness that  $\Phi_{v_1}=[i_{pq},y_1]$. It will become clear that the argument works for the other possible choices. Since, $T_b$ is a $(c_p,c_q)$-Tree, by induction, we can choose a labeling of $T_b$ with $\Phi_b=[i_{pq},y_b]$, such that $L(\Gamma,T_b)\geq L(\Phi,T_b)$. This gives
\begin{eqnarray}
 L(\Phi,T)&=&L(\Phi,T_b) +L(\Phi,T_1) + d_{\bm{\alpha}}[\Phi_{b},\Phi_{v_1}] \nonumber\\
 &=&L(\Phi,T_b) +L(\Phi,T_1) +\alpha_q[y_b,y_1]
 \end{eqnarray}
 Comparing the previous two equations with equation~\ref{eq:gamma}, we get,
 \begin{eqnarray}
L(\Gamma,T) &=&L(\Gamma,T_b) + L(\Gamma,T_1) + \alpha_p[x_b,x_1] + \alpha_q[y_b,y_1] \nonumber\\
&\geq& L(\Phi,T_b) +L(\Phi,T_1) + d_{\bm{\alpha}}[\Gamma_{v_1},\Phi_{v_1}] + \alpha_p[x_b,x_1] + \alpha_q[y_b,y_1]\nonumber\\
&=& L(\Phi,T_b) +L(\Phi,T_1) +\alpha_p[x_1,i_{pq}] + \alpha_p[x_b,x_1] + \alpha_q[y_b,y_1]\nonumber\\
&\geq&  L(\Phi,T_b) +L(\Phi,T_1) + \alpha_p[x_b,i_{pq}] + \alpha_q[y_b,y_1]\nonumber\\
&=&  L(\Phi,T_b) +L(\Phi,T_1) +  d_{\bm{\alpha}}[\Gamma_{b},\Phi_{b}] + d_{\bm{\alpha}}[\Phi_{b},\Phi_{v_1}]\nonumber\\
&=& L(\Phi,T) + d_{\bm{\alpha}}[\Gamma_{b},\Phi_{b}]
\end{eqnarray}
which satisfies the first possibility of the claim. It should be clear that if $\Phi_{v_1}=[x_1,i_{qp}]$ then the choice $\Phi_b=[x_b, i_{qp}]$  would give an identical bound.
\item $T_1$ is  a $(c_p,c_q)$-Tree, but $T_b$ is not. This case is similar to the previous one. Since $T_b$ has less than $n$ branch points, which are all outside $B_{pq}$, and it is not a $(c_p,c_q)$-Tree, we have from induction a labeling $\Phi$ of $T_b$ with either $\Phi_{b}=[i_{pq},y_b]$ and/or  $\Phi_{b}=[x_b,i_{pq}]$ such that
\begin{equation}
L(\Gamma,T_b)\ge L(\Phi,T_b) + d_{\bm{\alpha}}[\Gamma_b,\Phi_b]
\end{equation}
As before, let us assume $\Phi_{b}=[i_{pq},y_b]$ for concreteness. Since $T_1$ is a $(c_p,c_q)$-Tree, we can choose a labeling with $\Phi_{v_1}=[i_{pq},y_1]$ such that $L(\Gamma,T_1) \geq L(\Phi,T_1) $. This gives,
\begin{eqnarray}
 L(\Phi,T)&=&L(\Phi,T_b) +L(\Phi,T_1) + d_{\bm{\alpha}}[\Phi_{b},\Phi_{v_1}] \nonumber\\
 &=&L(\Phi,T_b) +L(\Phi,T_1) +\alpha_q[y_b,y_1]
 \end{eqnarray}
 Comparing the previous two equations with equation~\ref{eq:gamma}, we get,
\begin{eqnarray}
L(\Gamma,T) &=&L(\Gamma,T_b) + L(\Gamma,T_1) + \alpha_p[x_b,x_1] + \alpha_q[y_b,y_1] \nonumber\\
&\geq& L(\Phi,T_b) +L(\Phi,T_1) + d_{\bm{\alpha}}[\Gamma_{b},\Phi_{b}] + \alpha_p[x_b,x_1] + \alpha_q[y_b,y_1]\nonumber\\
&\geq& L(\Phi,T_b) +L(\Phi,T_1) + d_{\bm{\alpha}}[\Gamma_{b},\Phi_{b}]  +  d_{\bm{\alpha}}[\Phi_{b},\Phi_{v_1}] \nonumber\\
&=& L(\Phi,T) + d_{\bm{\alpha}}[\Gamma_{b},\Phi_{b}]
\end{eqnarray}
An identical argument carries through if $\Phi_b=[x_b,i_{qp}]$.
\item Neither $T_1$ or $T_b$ are $(c_p,c_q)$-Trees. It follows from induction that there is a labeling $\Phi$ such that $L(\Gamma,T_b)\ge L(\Phi,T_b) + d_{\bm{\alpha}}[\Gamma_b,\Phi_b]  $ and  $L(\Gamma,T_1) \geq L(\Phi,T_1) + d_{\bm{\alpha}}[\Gamma_{v_1},\Phi_{v_1}]  $. There are two possibilities in this case.
\begin{enumerate}
\item $(\Phi_b=[i_{pq},y_b]$ and $\Phi_{v_1}=[i_{pq},y_1])$ or $(\Phi_b=[x_b,i_{qp}]$ and $\Phi_{v_1}=[x_1,i_{qp}])$. As before, we will prove the claim for the former possibility while the later case can be proved by an identical argument.
    \begin{eqnarray}
 L(\Phi,T)&=&L(\Phi,T_b) +L(\Phi,T_1) + d_{\bm{\alpha}}[\Phi_{b},\Phi_{v_1}] \nonumber\\
 &=&L(\Phi,T_b) +L(\Phi,T_1) +\alpha_q[y_b,y_1]
\end{eqnarray}
\begin{eqnarray}
L(\Gamma,T) &=&L(\Gamma,T_b) + L(\Gamma,T_1) + \alpha_p[x_b,x_1] + \alpha_q[y_b,y_1] \nonumber\\
&\geq& L(\Phi,T_b) +L(\Phi,T_1) + d_{\bm{\alpha}}[\Gamma_{b},\Phi_{b}] + d_{\bm{\alpha}}[\Gamma_{v_1},\Phi_{v_1}] \nonumber\\
&+& \alpha_p[x_b,x_1] + \alpha_q[y_b,y_1]\nonumber\\
&\geq& L(\Phi,T_b) +L(\Phi,T_1) + d_{\bm{\alpha}}[\Gamma_{b},\Phi_{b}]  +  \alpha_q[y_b,y_1] \nonumber\\
&=& L(\Phi,T_b) +L(\Phi,T_1) + d_{\bm{\alpha}}[\Gamma_{b},\Phi_{b}]  +  d_{\bm{\alpha}}[\Phi_{b},\Phi_{v_1}] \nonumber\\
&=& L(\Phi,T) + d_{\bm{\alpha}}[\Gamma_{b},\Phi_{b}]
\end{eqnarray}
 This also satisfies the claim. The proof for $\Phi_b=[x_b,i_{qp}]$ and $\Phi_{v_1}=[x_1,i_{qp}]$ is identical.
 \item  $(\Phi_b=[i_{pq},y_b]$ and $\Phi_{v_1}=[x_1,i_{qp}])$ or $(\Phi_b=[x_b,i_{qp}]$ and $\Phi_{v_1}=[i_{pq},y_1])$. As before, we show the calculation for the former possibility. In this case 
 \begin{eqnarray}
 L(\Phi,T)&=&L(\Phi,T_b) +L(\Phi,T_1) + d_{\bm{\alpha}}[\Phi_{b},\Phi_{v_1}] \nonumber\\
 &=&L(\Phi,T_b) +L(\Phi,T_1) +\alpha_p[x_b,i_{pq}]+\alpha_q[i_{qp},y_1]
\end{eqnarray}
Combining this with equation~\ref{eq:gamma} we get,
 \begin{eqnarray}
L(\Gamma,T) &=&L(\Gamma,T_b) + L(\Gamma,T_1) + \alpha_p[x_b,x_1] + \alpha_q[y_b,y_1] \nonumber\\
&\geq& L(\Phi,T_b) +L(\Phi,T_1) + d_{\bm{\alpha}}[\Gamma_{b},\Phi_{b}] + d_{\bm{\alpha}}[\Gamma_{v_1},\Phi_{v_1}] \nonumber\\
&+& \alpha_p[x_b,x_1] + \alpha_q[y_b,y_1]\nonumber\\
&=& L(\Phi,T_b) +L(\Phi,T_1) + \alpha_p[x_b,i_{pq}] + \alpha_q[i_{qp},y_1] \nonumber\\
&+& \alpha_p[x_b,x_1] + \alpha_q[y_b,y_1]\nonumber\\
&\geq&L(\Phi,T_b) +L(\Phi,T_1) + \alpha_p[x_b,i_{pq}] + \alpha_q[i_{qp},y_1] \nonumber\\
&=&L(\Phi,T_b) +L(\Phi,T_1) + d_{\bm{\alpha}}[\Phi_{b},\Phi_{v_1}] =  L(\Phi_b,T)
 \end{eqnarray}
 But if we now relabel $b$ and $v_1$ with ${\tilde\Phi_{v_1}}= [i_{pq},i_{qp}]$ and  ${\tilde\Phi_{b}}= [i_{pq},i_{qp}]$ while ${\tilde\Phi_{v}}=\Phi_{v}$ for all other $v$, we get $L(\Phi,T_1) + \alpha_q[y_1,i_{qp}]\geq L({\tilde{\Phi}}_{v_1},T_1) $ and $L(\Phi,T_b) + \alpha_p[x_b,i_{pq}]\geq L({\tilde{\Phi}},T_b) $. This immediately gives,
 \begin{eqnarray}
 L({\tilde{\Phi}},T) &=& L({\tilde{\Phi}},T_b) + L({\tilde{\Phi}},T_1) + d_{\bm{\alpha}}[{\tilde{\Phi}}_{b}, {\tilde{\Phi}}_{v_1}] \nonumber\\
 &\geq&L(\Phi,T) \geq L(\Gamma,T)
 \end{eqnarray}
 Identical arguments work for the choices ${\tilde\Phi_{v_1}}= [x_1,i_{qp}]$ and ${\tilde\Phi_{b}}= [x_b,i_{qp}]$.
\end{enumerate}
\end{enumerate}
  This proves that if either of the two possibilities claimed are always true for an X-Tree with $n$ branch points or less then they are also true for a tree with $n+1$ branch points. The proof for arbitrary $n$ follows from induction.
\qed
\end{proof}

\begin{corollary} Given a minimum length phylogenetic X-Tree $T(\psi)$ there is an optimal labeling for each branch point within $\mathcal{B}$.
\end{corollary}
\begin{proof}
 Lemma~\ref{BpqLem} establishes that for any minimum Steiner tree labeled by $\Gamma$ and any branch point $b\in T$ such that $\Gamma_b\notin B_{pq}$, an alternative optimal labeling $\Phi$ exists such that  $\Phi_b$ is inside the union of blocks
\begin{displaymath}\Lambda(\Gamma_b,p,q)\equiv [c_p(i_{pq})c_q(i_{qp})]\cup [c_p(i_{pq})c_q((\Gamma_{b})_q)]\cup [c_p((\Gamma_b)_p)c_q(i_{qp})] \end{displaymath} If we root the tree at $b$, the new optimal labeling for all its descendants is inferred in a backward pass of the Erdos-Szekely algorithm. This ensures that each branch point in a minimum length phylogenetic X-Tree is labeled inside $B_{pq}$.
Let $S_b =\cap_{B_{pq}\neq \mathcal{G}}\Lambda(\Gamma_b,p,q) \subseteq \mathcal{B}$,  where the intersection is taken over all pair of characters for which the pruning condition is satisfied. It follows from Lemma~\ref{BpqLem} that $S_b$ also contains an alternate optimal labeling of $T(\psi)$. Note that $S_b$ is a non-empty subset of $\mathcal{B}$. This must be true because given a character pair $c_p,c_q$, each union of blocks contains at least one taxon and so the rule matrix $R(p,q)$ that defines the Buneman graph must have ones for each of these blocks. Therefore each element in  $S_b$ represents a distinct vertex of the Buneman graph.
\qed
\end{proof}
As argued before, any minimum Steiner tree can be decomposed uniquely into phylogenetic X-Tree components and the previous corollary ensures that each phylogenetic X-Tree can be labeled optimally inside the generalized Buneman graph. It follows that all distinct minimum Steiner trees are contained inside the generalized Buneman graph.

\subsection*{Integer Linear Program (ILP) Construction}
We briefly summarize the ILP flow construction used to find the optimal phylogeny.  We convert the generalized Buneman graph into a directed graph by replacing an edge between vertices $u$ and $v$ with two directed edges $(u,v), (v,u)$ each with weight $w_{uv} $ as determined by the distance metric. Each directed edge has a corresponding binary variable $s_{u,v}$ in our ILP. We arbitrarily choose one of the taxa as the root $r$, which acts as a source for the flow model.  The remaining taxa $T\equiv \chi - \{r\}$ correspond to sinks.  Next, we set up real-valued flow variables $f_{u,v}^{t}$, representing the flow along the edge $(u,v)$ that is intended for terminal $t$. The root $r$ outputs $|T|$ units of flow, one for each terminal.  The Steiner tree is the minimum-cost tree satisfying the flow constraints. This ILP was described in \cite{Sri}, and we refer the reader to that paper for further details. The ILP for this construction of the Steiner tree problem is the following:

\begin{eqnarray}
\textrm{Minimize} \sum_{(u,v)\in\mathcal{B}}  w_{uv}s_{u,v} \nonumber \\
\textrm{~subject to}~\sum_v  (f_{u,v}^{t} - f_{v,u}^{t})=0&& ~\forall u\in \mathcal{B}\setminus\{t,r\},~\forall t\in T\nonumber \\
\sum_v (f_{r,v}^{t} - f_{v,r}^{t} )= 1&&~\forall t\in T\nonumber \\
0 \leq f_{u,v}^{t} \leq s_{u,v}&&~\forall (u,v)\in\mathcal{B},~\forall t\in T\nonumber \\
s_{u,v} \in \{0,1\}&&~\forall (u,v)\in\mathcal{B}
\end{eqnarray}

\section*{Results}

\begin{table}[t!]
\caption{Pruning and run time results for the data sets reported.}
\begin{center}
\resizebox{12.2cm}{!}{
\begin{tabular}{|c|c|c|c|cc|cc|cc|}

\hline
Data & Input & Complete & $|\mathcal{B}|$ & \multicolumn{2}{c|}{ILP} & \multicolumn{2}{c|}{{ \tt pars}} &\multicolumn{2}{c|}{{ \tt PAUP*}} \\
& (raw) & graph & & length & time  & length & time &length & time \\ \hline
{\em O. rufipogon} DNA & $41 \times 1043$ & $2^{18}*3^{2}$ & 58 & 57 & 0.29s & 57 & 2.57s& 57 & 2.09s \\ \hline
Human mt-DNA & $80\times 245$ & $2^{28}$ & 64 & 44 & 0.48s & 45 & 0.56s & 44 & 5.69s \\ \hline
HIV-1 RT protein& $50\times176$ & $2^{16}*3*4^{2}$ & 297 & 40 & 127.5s & 42 &0.30s& 40 &3.84s  \\ \hline \hline
mt3000 & $500 \times 3000$ & $2^{99}*3^{2}$ &322& 177& 40s& 178 & 2m37s & 177&5h23m \\ \hline
mt5000a & $500 \times 5000$ & $2^{167}*3^{2}$ &1180& 298 & 5h10m& 298 & 35m49s & 298 &3h52m \\ \hline
mt5000b & $500\times 5000$ & $2^{229}*3^{3}$ & 360 & 312& 3m41s & 312 & 57m6s & 312 & 2h40m \\ \hline
mt10000& $500\times10000$ & $2^{357}*3^{5}$ &6006& N. A. &N. A. & 637 & 1h34m& 637 &1h39m \\ \hline
\end{tabular}
}
\label{DataTable1}
\end{center}
\end{table}
We implemented our generalized Buneman pruning and the ILP in C++. The ILP was solved using the Concert callable library of CPLEX 10.0. We compared the performance of our method with two popular heuristic methods for maximum parsimony phylogeny inference --- {\tt pars}, which is part of the freely-available PHYLIP package~\cite{Phylip}, and {\tt PAUP*}~\cite{PAUP}, the leading commercial phylogenetics package. We attempted to use PHYLIP's exact branch-and-bound method {\tt DNA penny} for nucleotide sequences, but discontinued the tests when it failed to solve any of the data sets in under 24 hours.  In each case, {\tt pars} and {\tt PAUP*} were run with default parameters.  We first report results from three moderate-sized data sets selected to provide varying degrees of difficulty:  a set of 1,043 sites from a set of 41 sequences of {\em O.~rufipogon} (red rice)~\cite{rice},  245 positions from a set of 80 human mt-DNA sequences reported by~\cite{Ychr}, and 176 positions from 50 HIV-1 reverse transcriptase amino acid sequences. The HIV sequences were retrieved by NCBI BLASTP~\cite{Blast} searching for the top 50 best aligned taxa for the query sequence GI 19571541 and default parameters.  We then added additional tests on larger data sets all derived from human mitochondrial DNA. The mtDNA data was retrieved from NCBI BLASTN, searching for the 500 best aligned taxa for the query sequence GI 61287976 and default parameters. The complete set of 16,546 characters (after removing indels) was then broken in four windows of varying sizes and characteristics: the first 3,000 characters (mt3000), the first 5,000 characters (mt5000a), the next 5,000 characters (mt5000b), and the first 10,000 characters (mt10000).
Table~\ref{DataTable1} summarizes the results.

For the set of 41 sequences of lhs-1 gene  from {\em O.~rufipogon} (red rice)~\cite{rice}, our method pruned the full graph of  $2^{18}*3^{2}$ nodes (after screening out redundant characters) to 58.  Fig \ref{Phylogenies}(a) shows the resulting phylogeny.  Both {\tt PAUP*} and {\tt pars} yielded an optimal tree although more slowly than the ILP (2.09 seconds and 2.57 seconds respectively, as opposed to  0.29 seconds).

For the 245-base human mt-DNA sequences, the generalized Buneman pruning was again highly efficient, reducing the state set from $2^{28}$ after removing redundant sequences to 64.  Fig \ref{Phylogenies}(b) shows the phylogeny returned.  While {\tt PAUP*} was able to find the optimal phylogeny (although it was again slower at 5.69 seconds versus 0.48 seconds), {\tt pars} yielded a slightly sub-optimal phylogeny (length 45 instead of 44) in a comparable run time (0.56 seconds).

For HIV-1 sequences, our method pruned the full graph of $2^{16}*3*4^{2}$ possible nodes to a generalized Buneman graph of 297 nodes, allowing solution of the ILP in about two minutes. Fig \ref{Phylogenies}(c) shows an optimal phylogeny for the data.  {\tt PAUP*} was again able to find the optimal phylogeny and in this case was faster than the ILP (3.84 seconds as opposed to 127.5 seconds). {\tt pars} required a shorter run time of 0.30 seconds, but yielded a sub-optimal tree of length of 42, as opposed to the true minimum of 40.

\begin{figure}[t!]
\begin{center}
\resizebox{12.2cm}{!}{
\includegraphics[scale=0.5]{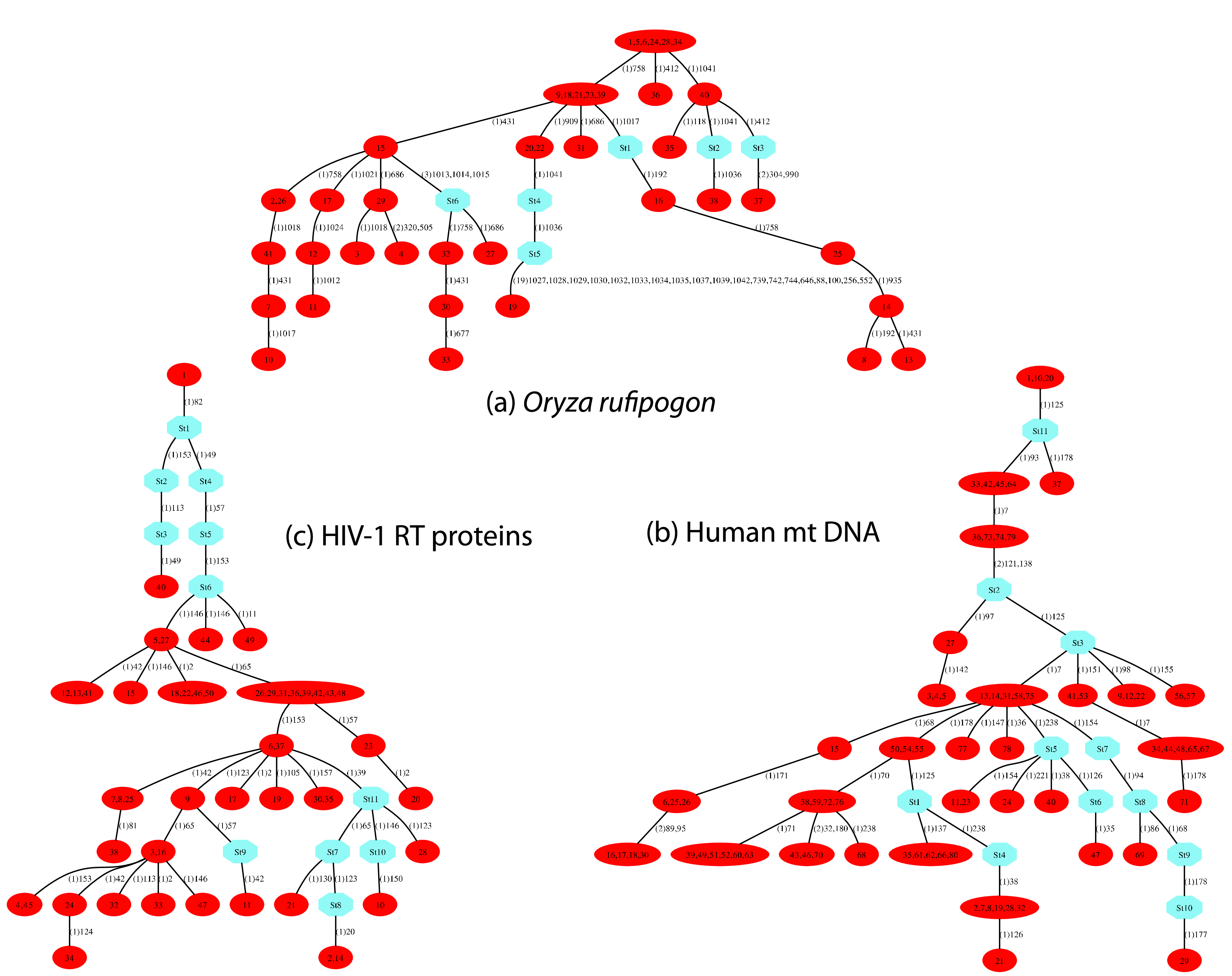}
}
\caption{Most parsimonious phylogenies (a) lhs-1 gene for {\em O.~rufipogon} \cite{rice} (b) Human mt-DNA \cite{Ychr} and (c)  HIV-1 RT proteins \cite{Blast}. Edges are labelled by their lengths in parentheses followed by sites that mutate along that edge. Dark red ovals are input taxa and light blue Steiner nodes.}
\label{Phylogenies}
\end{center}
\end{figure}

For the four larger mitochondrial datasets, Buneman pruning was again highly effective in reducing graph size relative to the complete graph, although the ILP approach eventually proves impractical when Buneman graph sizes grows sufficiently large.  Two of the data sets yielded Buneman graphs of size below 400, resulting in ILP solutions orders of magnitude faster than the heuristics. mt5000a, however, yielded a Buneman graph of over 1,000 nodes, resulting in an ILP that ran more slowly than the heuristics.  mt10000 resulted in a Buneman graph of over 6,000 nodes, leading to an ILP too large to solve.  {\tt pars} was faster than {\tt PAUP*} in all cases, but {\tt PAUP*} found optimal solutions for all three instances we can verify while {\tt pars} found a sub-optimal solution in one instance.

We can thus conclude that the generalized Buneman pruning approach developed here is very effective at reducing problem size, but solving provably to optimality does eventually become impractical for large data sets.  Heuristic approaches remain a practical necessity for such cases even though they cannot guarantee, and do not always deliver, optimality.  Comparison of {\tt PAUP*} to {\tt pars} and the ILP suggests that more aggressive sampling over possible solutions by the heuristics can lead optimality even on very difficult instances but at the cost of generally greatly increased run time on the easy to moderate instances.

\section*{Discussion}
We have presented a new method for finding provably optimal maximum parsimony phylogenies on multi-state characters with weighted state transitions, using integer linear programming.  The method builds on a novel generalization of the Buneman graph for characters with arbitrarily large but finite state sets and for arbitrary weight functions on character transitions.  Although the method has an exponential worst-case performance, empirical results show that it is fast in practice and is a feasible alternative for data sets as large as a few hundred taxa. While there are many efficient heuristics for recontructing maximum parsimony phylogenies, our results cater to the need for provably exact methods that are fast enough to solve the problem for biologically relevant multi-state data sets. Our work could potentially be extended to include more sophisticated integer programming techniques that have been successful in solving large instances of other hard optimization problems, for instance the recent solution of the 85,900-city traveling salesman problem pla85900~\cite{Applegate}. The theoretical contributions of this paper may also prove useful to work on open problems in multi-state MP phylogenetics, to accelerating methods for related objectives, and to sampling among optimal or near-optimal solutions.

\section*{Acknowledgements}
NM would like to thank Ming-Chi Tsai for several useful discussions. This work was supported in part by NSF grant \#0612099.

\end{document}